\documentclass[10pt,journal,cspaper,compsoc]{IEEEtran}

\usepackage{graphicx}
\usepackage{amsmath}
\usepackage{enumerate}
\usepackage{amssymb,amsfonts}
\usepackage{cases}
\usepackage{algorithm}
\usepackage{algorithmic}
\usepackage{amsthm}
\usepackage{graphicx} \usepackage{subfigure}

\usepackage{cite}
\usepackage{subfigure}
\usepackage{epstopdf}
\usepackage{dsfont}
\usepackage{bm}
\usepackage{subeqnarray}
\usepackage{multirow}
\usepackage{slashbox}
\usepackage{amscd}
\usepackage{balance}

\makeatletter
\newcommand{\rmnum}[1]{\romannumeral #1}
\newcommand{\Rmnum}[1]{\expandafter\@slowromancap\romannumeral #1@}
\makeatother

\newtheorem{remark}{Remark}
\newtheorem{theorem}{\bf Theorem}
\newtheorem{lemma}{\bf Lemma}
\newtheorem{definition}{\bf Definition}

\begin{document}


\title{Group Key Agreement Protocol for MANETs Based on HSK Scheme}
\author{Xinyu Lei,~Xiaofeng Liao~\IEEEmembership{Senior Member,~IEEE}, and Yonghong Xiong
\IEEEcompsocitemizethanks{\IEEEcompsocthanksitem Xinyu Lei, Xiaofeng Liao, and Yonghong Xiong are with the State Key Lab. of Power Transmission Equipment \& System Security and New Technology, College of Computer Science, Chongqing University, Chongqing 400044, P. R. China. E-mail: xy-lei@qq.com, xfliao@cqu.edu.cn, xiongyonghong\_cqu@163.com.

}
\thanks{}}

\markboth{}
{Lei \MakeLowercase{\textit{et al.}}: Group Key Agreement Protocol for MANETs Based on HSK Scheme}
\IEEEpubid{}

\IEEEcompsoctitleabstractindextext{%
\begin{abstract}
In this paper, we first provide a spanning tree (ST)-based centralized group key agreement protocol for unbalanced mobile Ad Hoc networks (MANETs). Based on the centralized solution, a local spanning tree (LST)-based distributed protocol for general MANETs is subsequently presented. Both protocols follow the basic features of the HSK scheme: 1) H means that a hybrid approach, which is the combination of key agreement and key distribution via symmetric encryption, is exploited; 2) S indicates that a ST or LSTs are adopted to form a connected network topology; and 3) K implies that the extended Kruskal algorithm is employed to handle dynamic events. It is shown that the HSK scheme is a uniform approach to handle the initial key establishment process as well as all kinds of dynamic events in group key agreement protocol for MANETs. Additionally, the extended Kruskal algorithm enables to realize the reusability of the precomputed secure links to reduce the overhead. Moreover, some other aspects, such as the network topology connectivity and security, are well analyzed.
\end{abstract}

}

\maketitle

\IEEEdisplaynotcompsoctitleabstractindextext

\IEEEpeerreviewmaketitle

\section{Introduction}

\IEEEPARstart{M}obile Ad Hoc networks (MANETs) are currently deployed in many areas ranging from military, emergency, rescue mission to other collaborative applications for commercial use. The rapidly growing application of MANETs heightens the security concerns of such networks. Designing secure and efficient group key agreement protocols for MANETs has attracted significant attention. A group key agreement protocol enables a group of participants to communicate over untrusted, open networks to come up with a common secret value called a \emph{session key}. Theoretically, group key establishment is more efficient than pairwise key establishment as the communicating nodes do not waste resources every time they wish to communicate with another device.

Most group key agreement protocols are based on generalization of the two-party Diffie-Hellman (DH) \cite{diffie1976new} protocol. It is proposed by Imgemarsson et al. \cite{ingemarsson1982conference} the first group key agreement protocol. A more efficient BD protocol, which requires only two rounds by using a ring structure of participants, is introduced in \cite{burmester1995secure}.
Subsequently, Steiner et al. \cite{steiner1996diffie} provides a serial of GDH protocols which are called the natural extension of the original DH protocol.
Hypercube protocol is introduced in \cite{becker1998communication} in which the participants are arranged in a logical hypercube. Beside the above \emph{static protocols}, particular attention is paid to design \emph{dynamic protocols} which are capable of handling the membership changes. As a typical example of dynamic protocols, TGDH \cite{kim2004tree} protocol employs tree structure, in which any node can compute the group key if it knows all the keys in its co-path. Unfortunately, this requirement makes the protocol quite expensive in storage and computation. Barus et al. \cite{barua2003extending} presents BDS protocol which is also based on tree structure. It is the only one that makes use of the parings and ternary trees. There have been a tremendous amount of dynamic group key agreement protocols \cite{kim2004constant, bresson2002dynamic, bresson2004mutual, dutta2005constant, kim2000simple, kim2001communication, steiner1998cliques, sherman2003key, dutta2008provably, steiner2000key}. One may refer to \cite{dutta2005overview} for a detailed survey.

Most of the above traditional dynamic protocols are efficient for wired networks but they cannot be directly applied to MANETs because of dynamic and multi-hop nature of mobile nodes. Indeed, most protocols for wired networks have assumed that the communication cost between any pair of group members is one unit, which is not practical in wireless multi-hop MANETs. In contrast, unlike wired networks, the nodes in a MANETs are energy constrained (often powered by batteries) with limited storage and the usage of wireless communication implies a limited bandwidth. As a result, the nodes in MANETs have limited communication, computation, and storage ability. In addition, the global broadcast in MANETs is probably infeasible in MANETs. The membership changes rapidly due to the high mobility of nodes. Moreover, the mobile nodes lack pre-shared knowledge, and the multi-hop wireless links are vulnerable. As a consequence, it remains challenging to design group key agreement protocol in MANETs.
\subsection{Related Work}

In recent few years, some group key agreement protocols using certain structures, e.g., tree or cluster, have shown their superiorities in designing group key agreement protocol in MANETs. They are referred to as \emph{hierarchical protocols} in this paper.

Regarding to the tree-based protocols, a ST network structure is employed in AT-GDH protocol \cite{hietalahti2001efficient}, which realizes that all communications are based on one-hop transmission. But the protocol requires too much exponentiations to derive a common session key which is not computationally efficient in the mobile environment.

The seminal cluster-based group key agreement protocol is first proposed in \cite{li2002efficient}, and then further developed in \cite{abdel2007authenticated, hietalahti2008clustering, konstantinou2008cluster, yao2003making, dutta2011provably}.
The concept of \emph{connected dominating set} is introduced in  \cite{li2002efficient} to handle key agreement protocol. However, the protocol is inefficient to handle dynamic events. There are a few proposals for key agreement protocols by means of using pairing  \cite{abdel2007authenticated, konstantinou2008cluster, shi2006authenticated}. The protocol in \cite{shi2006authenticated} exploit the concept of \emph{virtual backbone node} with leaf nodes as participants. Unlike \cite{shi2006authenticated}, the internal nodes in \cite{abdel2007authenticated, konstantinou2008cluster} are also participants. The pairing-based protocols, however, are computation-intensive which impedes their application in MANETs. Based on AT-GDH, it is presented in \cite{hietalahti2008clustering} a solution which employs BD protocol within each cluster and invokes AT-GDH protocol by employing a ST of sponsors. Unfortunately, this protocol is static and handling dynamic events are not easy for this scheme. Impressively, a hybrid HP-1 protocol is presented in \cite{dutta2011provably}. After the cluster construction, the BD algorithm is invoked to generate the session key within the cluster, and then the root node generates a random session key and sends it to other nodes via the established secure links. It is shown that the HP-1 protocol is very efficient, since it makes use of a hierarchical structure to manage the dynamic events, and employs the symmetric encryption algorithm to distribute the session key from the root node to other nodes. Different with the previous work, the HP-1 protocol takes advantage of the symmetric encryption algorithm which is regarded to be more efficient than the asymmetric encryption algorithm. A detailed survey of the cluster-based algorithms can be found in \cite{yu2005survey}.

After having an overview of these existing hierarchical protocols, especially the cluster-based protocols, main notable merits are identified as follows.
\begin{itemize}
\item A hierarchical structure is adopted to handle the dynamic events efficiently.
\item A hybrid encryption is employed. This approach can reduce the computation overhead, and therefore, it is quite suitable for MANETs with resource-constrained mobile nodes.
\end{itemize}

We also identifies some common drawbacks in these existing hierarchical protocols.
\begin{itemize}
\item	Clustering method is not easy to handle certain member events, such as a cluster head node leaving from the network. More precisely, it is rather costly to use clustering method to deal with the situation that several cluster head nodes leaving from the network at the same time.
\item	Distinct complex algorithms should be carefully designed for handling different kinds of dynamic events. Intuitively, a better approach is to use one uniform algorithm to deal with all these events.
\end{itemize}

In this paper, we aim to design a group key agreement scheme which inherits the above merits and overcomes the above drawbacks.

\subsection{Contributions}

We regard the main contributions of this paper as five-fold.
\begin{enumerate}
\item	The weight function is introduced which jointly considers high communication efficiency and energy balance.
\item	Hybrid encryption technique applies the symmetric encryption algorithm which is regarded to be more efficient than asymmetric encryption algorithm.
\item	The extended Kruskal algorithm is very efficient to handle dynamic events. All kinds of dynamic events are well addressed by only one uniform algorithm, whereas other approaches always involve in designing several sophisticated algorithms to handle different kinds of dynamic events, e.g., \cite{dutta2011provably}.
\item	The extended Kruskal algorithm enables to realize the reusability of the precomputed secure links, and thus reduces the overhead.
\item	There is no global broadcast in the proposed protocols. That is to say, all transmissions are based on one-hop transmission.
\end{enumerate}

\subsection{Organization}
The remainder of the paper proceeds as follows.
Section 2 provides the centralized spanning tree (ST)-based protocol. The distributed local spanning tree (LST)-based protocol is described in Section 3. In Section 4, some properties of the network topology are theoretically studied. The protocol properties and security are analyzed in Section 5 and Section 6, followed by Section 7 which summarizes the paper.

\section{The Centralized ST-based Protocol}

In hierarchical group key agreement protocols, the management of the topology, tree or cluster, can be either centralized or distributed. In the first approach, a member of the group is responsible for the topology maintenance. While the second approach involves all group members in the topology maintenance. Remarkably, the words "centralized" and "distributed" are in terms of the topology maintenance process.
In this section, a centralized solution for a special kind of MANETs is constructed. This solution
 paves a way for better understanding the distributed solution for general MANETs.

In some applications, MANETs can be regarded as unbalanced networks which consist of many resource-constrained mobile nodes called \emph{normal nodes} and a powerful node called \emph{leader node} with less restriction. We refer to this kind of MANETs as \emph{unbalanced} MANETs. Shown in Fig. \ref{fig10} is an example of unbalanced MANETs,
where a tank serves as a leader node and many soldiers work as normal nodes to collect useful battlefield information and transmit it to the tank. Another example is urban vehicular ad hoc networks. It can be regarded as unbalanced MANETs where city central base station works as a leader node and many vehicles act as normal nodes.

We formalize the assumptions of unbalanced MANETs as follows.
\begin{enumerate}[\emph{A}1)]
\item	The leader node covers all normal nodes.
\item	All normal nodes are homogeneous with an identical transmission range.
\end{enumerate}
According to Assumption \emph{A}1, the leader node can transmit message to all normal nodes directly, whereas that normal nodes transmit message to the leader node should rely on mediate nodes. Mathematically, the transmission range of nodes is defined as
\begin{equation*}\label{eq2-20}
D_i\!=\!
\begin{cases}
d_{\mathrm{leader}}\geq \mathrm{max}\{d_{11},\ldots,d_{1n}\},~\mathrm{if}~i=1 \cr
d_{\mathrm{normal}},~\mathrm{if}~i=2,\ldots, n
\end{cases}
\end{equation*}
where $d_{ij}$ denotes the physical distance between node $v_i$ and node $v_j$. $V$ is the set of participant nodes in network, $n=|V|$ is the number of nodes in $V$, and  $(v_i, v_j)$ is the edge that connects node $v_i$ and node $v_j$. A unique ID is assigned to each node. For notational simplicity, $ID(v_i)=i$.

\begin{figure}[t]
\centering
\includegraphics[width=4 cm]{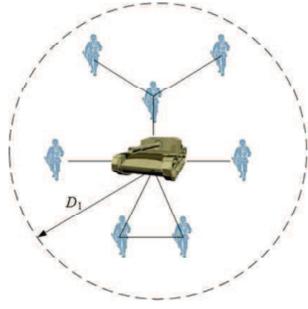}
\caption{An example of unbalanced MANETs.} \label{fig10}
\end{figure}
\begin{definition}(Directly Connected Relationship):
Node $v_i$ has directly connected relationship with node $v_j$, if node $v_i$ can directly transmit message to node $v_j$, i.e., $d_{ij}<D_i$ , denoted as $v_i\rightarrow v_j$. It follows that $v_i\leftrightarrow v_j$, if and only if $v_i\rightarrow v_j$  and  $v_i\leftarrow v_j$.
\end{definition}

Definition 1 and Assumption \emph{A}1 yield $v_1\rightarrow v_i$, $\forall i\in\{2,\ldots,n\}$. And it follows that  $v_i\rightarrow v_j$, where $i\neq j$ and $i,j>1$, if and only if $v_i\leftarrow v_j$. That is to say, the connection between two normal nodes must be bidirectional.

\subsection{Fundamental Procedures}
\subsubsection{Topology Construction}
Normal nodes own the ability of \emph{local broadcast}, i.e., a node can broadcast information to its one-hop away neighbor nodes.

\begin{algorithm}[h]
\caption{\textbf{Procedure T-C} (\emph{Topology Construction})}
\begin{algorithmic}[1]

\STATE Each node locally broadcasts a \texttt{Hello Message}, which contains its unique ID;
\STATE Each node sends an \texttt{ID Message} which contains its ID and its neighbors' IDs to the leader node. In this way, the leader node constructs the network topology.
\end{algorithmic}
\end{algorithm}

\subsubsection{Design and Calculation of Weight Function}
A weight function is introduced to jointly consider high communication efficiency and energy balance. For this purpose, normal node should send extra information (called a\texttt{ Weight Message}), which contains the position and power available of the node, to the leader node. As an example, the weight function can be designed as
\begin{align}
W_{ij}&=f(d_{ij}, PA_i, PA_j)=\notag\\
&\begin{cases}
+\infty,~\mathrm{if}~d_{ij}> d_{\mathrm{normal}}, \cr
M+\alpha\cdot d_{ij}-\beta\cdot\mathrm{min}\{PA_i, PA_j\},~\mathrm{if}~d_{ij}\leq d_{\mathrm{normal}},\notag
\end{cases}
\end{align}
where $W_{ij}$ represents the weight of edge $(v_i, v_j)$, and $PA_i$ ($PA_j$) represents the power available of node $v_i$ ($v_j$). $\alpha$ and $\beta$ are the ratio parameters. A large constant $M$ is selected, such that $M+\alpha\cdot d_{ij}-\beta\cdot\mathrm{min}\{PA_i, PA_j\}>0$.

Note that the smaller $W_{ij}$ of an edge, the higher probability it will be selected in the ST (see the extended Kruskal algorithm for the detailed reason). If $d_{ij}>d_{\mathrm{normal}}$, or equivalently, $v_i$  and $v_j$  are disconnected in network topology, set $W=+\infty$. If $d_{ij}<d_{\mathrm{normal}}$, to construct a communication efficient ST, $W_{ij}$ should be an increasing function of $d_{ij}$, as communication power consumption is, in general, of the form $c\cdot d^r$, $r>2$, i.e., a strictly increasing function of the Euclidean distance. On the other hand, $W_{ij}$ should be a decreasing function of $\mathrm{min}\{PA_i, PA_j\}$. As a result, an edge that connects two energy-rich nodes has higher probability to be selected in the ST. There are three points worth mentioning here. The first is the ratio parameters, i.e., $\alpha$  and $\beta$, should be carefully determined according to the requirements of practical applications. Second, we only do a qualitative analysis of how to design the weight function $f$  by considering only two factors. It is notable that, if required, other factors can also be taken into consideration. The third one is that, if necessary, the weight function can be designed as any nonlinear function which follows our aforementioned qualitative analysis. The procedure to calculate weight for each edge is described below.

\begin{algorithm}[h]
\caption{\textbf{Procedure C-W} (\emph{Calculate Weight})}
\begin{algorithmic}[1]
\STATE Each normal node sends a \texttt{Weight Message} to the leader node;
\STATE The leader node calculates the weight for each edge according to a predefined weight function $f$.
\end{algorithmic}
\end{algorithm}

\subsubsection{ST Construction}
After the execution of \textbf{Procedure T-C} and \textbf{Procedure C-W}, the leader node can obtain the network topology and weight of each edge, on the basis of which the leader node can construct a ST by using Kruskal algorithm \cite{kruskal1956shortest} or Prim algorithm \cite{prim1957shortest}. However, the topology of the constructed ST is unknown to normal nodes in the network. Given this, the ST needs to be further constructed among these normal nodes in the network. We assume that the leader node has constructed a ST. Below is the procedure to construct the ST among the normal nodes.

\begin{algorithm}[h]
\caption{\textbf{Procedure S-C }(\emph{ST Construction})}
\begin{algorithmic}[1]
\STATE The leader node transmits neighbors' IDs (called a \texttt{Notification Message}) to each normal node in the constructed ST. In this way, each normal node is aware of which nodes are its neighbors in the constructed ST;
\STATE For each edge in the ST, two nodes perform a two-party public key exchange algorithm, e.g., DH protocol \cite{diffie1976new}, to generate a pair of key. Evidently, an edge in the ST implies a pairwise public key establishment process, and all of them are executed simultaneously.
\end{algorithmic}
\end{algorithm}

The edges which have performed the public exchange algorithm are referred as \emph{secure links}, and the generated key is referred as \emph{secure link key}. When the pairwise public key establishment is finished, two adjacent nodes in the ST can communicate by using a symmetric encryption algorithm with the secure link key. If we treat the leader node as the root node, then the ST forms a hierarchical tree structure, in which the internal nodes have secure links with their father nodes and their children nodes, and the leaf nodes only have secure links with their father nodes.

\subsubsection{Session Key Distribution}
After the execution of \textbf{Procedure S-C}, a ST has been constructed among all normal nodes. Next, we describe \textbf{Procedure S-K-D}.

\begin{algorithm}[h]
\caption{\textbf{Procedure S-K-D} (\emph{Session Key Distribution})}
\begin{algorithmic}[1]
\STATE The leader node generates a random session key, drawn uniformly from the key space, and sends to its children in the ST encrypted by their secure link key using any symmetric algorithm, e.g., AES or DES (see \cite{menezes2010handbook});
\STATE	All of the one-hop children of the leader node in the first level of the ST decrypt the encrypted message and recover the session key;
\STATE	Each child node reencrypts the recovered session key with the secure link key of its children, and then sends the encrypted message to its children;
\STATE	Repeat Step 3 until all of the leaf nodes in the ST are reached.

\end{algorithmic}
\end{algorithm}

\subsubsection{Topology Maintenance}
\begin{definition}
(Node Event, denoted as $\eta$):
\begin{enumerate}
\item	A Node Adding Event occurs when a single or a group of nodes join the network topology;
\item	A Node Deleting Event occurs when a single or a group of nodes leave the network topology or are forced to leave.
\end{enumerate}
\end{definition}
Different with some related literatures, we unify the definition of a single node adding (deleting) event and a group node adding (deleting) event. With this definition, a membership change can represented by a node event.

\begin{definition}
(Edge Event, denoted as $\varepsilon$):
\begin{enumerate}
\item	An Edge Adding Event occurs when a single or a group of edges are added to the network topology;
\item	An Edge Deleting Event occurs when a single or a group of edges are deleted from the network topology.
\end{enumerate}
\end{definition}
Frequently, an edge event arises from the high mobility of the nodes. The session key should be updated if any node event occurs. But it is not necessary to update the session key if only edge event occurs.

\begin{definition}(Topology $G_k$): The topology, $G_k(k=0,1,\ldots)$, is an undirected graph $G_k=(V_k, G_k) $, where $V_k$ is all the nodes in the network after the occurrence of $k$th node event, and $E_k=\{(v_i,v_j):v_i\leftrightarrow v_j\}$. In particular, $G_0$ is the initial network topology with no node event occurs.
\end{definition}

Our approach to maintain a connected ST is inspired by the basic idea of the original Kruskal algorithm. It is referred to as \emph{extended Kruskal algorithm}. After the usage of extended Kruskal algorithm in topology $G_k$, a connected ST $G_k^+$ can be derived. Hence, we have the following definition.

\begin{definition}(Topology $G_k^+$):
The topology, $G_k^+(k=0,1,\ldots)$, is a connected ST  $G_k^+=(V_k^+,E_k^+)$, where $V_k^+=V_k$, and  $E_k^+$ denotes all the secure links which have been constructed after the extended Kruskal algorithm is applied.
\end{definition}

\begin{definition}(Topology $G_k^-$):
The topology,  $G_k^-(k=0,1,\ldots)$, is an undirected graph $G_k^-=(V_k^-,E_k^-)$, where  $V_k^-=V_k$, and $E_k^-$ denotes all the secure links which have been constructed in $G_{k-1}^+$ . In particular, $G_0^-$ is a graph with isolated nodes.
\end{definition}

With the above terminologies, the evolution of the network topology can be formally described as: after the occurrence of $k$th node event, the leader node reconstructs topology $G_k$  and $G_k^-$. Generally, $G_k^-$ is a separated graph. Then, the extended Kruskal algorithm is employed to construct a new connected ST $G_k^+$ based on $G_k^-$. The procedure to construct $G_k^+$ based on $G_k^-$ is elaborated in \textbf{Procedure T-M}.

\begin{algorithm}[h]
\caption{\textbf{Procedure T-M} (\emph{Topology Maintenance}) (\emph{The Extended Kruskal Algorithm})}
\begin{algorithmic}[1]
\STATE If $G_k^-$  is disconnected, without loss of generality, we assume that $G_k^-$ is separated into $p$  partial trees. Based on $p$  partial trees, in each step, select the minimum previously unselected edge to connect two separated partial trees, two partial trees merge into one.

\STATE Repeat the above step until there is only one partial tree. And finally, we get a new ST $G_k^+$.

\end{algorithmic}
\end{algorithm}

\begin{remark}

\begin{enumerate}
\item The original Kruskal algorithm starts with $n$ isolated nodes, each node is treated as a separated partial tree. The above algorithm is referred to as the extended Kruskal algorithm because it may start with a graph where several edges have been selected, and each connected component is treated as a partial tree.
\item The extended Kruskal algorithm unifies the processes of constructing the initial ST and maintaining a ST. In particular, $k=0$ denotes the initial key establishment process. In the initial key establishment process, the leader node constructs $G_0$ and $G_0^-$ . $G_0^-$ is treated as a graph with isolated nodes, in which case the extended Kruskal algorithm reduces to the original Kruskal algorithm, and  $G_0^-$ is exactly a minimum ST (MST) in $G_0$.

\item The extended Kruskal algorithm preserves the recomputed secure links in $G_k^+$, and thus reduces the computation and communication overhead to construct new secure links.
\end{enumerate}
\end{remark}

\subsection{The Completed Centralized Protocol}

\noindent\rule{255pt}{1pt}

\textbf{The Centralized Protocol}

\vspace{-2mm}

\noindent\rule{255pt}{1pt}
\vspace{-5mm}
\begin{enumerate}[Step 1:]

\item After the occurrence of $k$th node event $\eta_k (k=0,1,\ldots)$, the leader node invokes \textbf{Procedure T-C} to construct the topology $G_k$  and $G_j^-$;

\item The leader node invokes \textbf{Procedure C-W} to update the weight for each edge in $G_k^-$;

\item If $G_k^-$ is connected, then set $G_k^+=G_k^-$, go to Step 6;

\item	If $G_k^-$ is disconnected, the leader node invokes \textbf{Procedure T-M} to construct a new ST  $G_k^+$ based on $G_k^-$ in $G_k$;
\item	The leader node invokes \textbf{Procedure S-C} to construct the ST $G_k^+$ among all normal nodes;

\item	Based on $G_k^+$, the leader node invokes \textbf{Procedure S-K-D} to update a session key.
\end{enumerate}
\vspace{-3mm}
\noindent\rule{255pt}{1pt}

\begin{remark}

\begin{enumerate}
\item    The weight is updated every time when a node event occurs, which implies that the weight is time-variant.
\item $k=0$ denotes the initial key establishment process, in which the extended Kruskal algorithm reduces to the original Kruskal algorithm. In this case, the leader node may alternatively employ Prim algorithm to construct a MST. It is shown that the Kruskal algorithm is appropriate for a sparse graph, whereas it is better to use the Prim algorithm for a dense graph.

\item $G_k^+(k=1,2,\ldots)$ is not necessarily a MST in $G_k$.

\item Between two node events, there may be several edge events which may also destroy the connectivity of $G_k^-$. Other existing hierarchical protocols are difficult to handle the occurrence of edge events. Fortunately, our algorithm is capable of addressing this situation efficiently. In \textbf{Procedure T-C}, the leader node considers both node events and edge events which may lead to the topology change, and then reconstructs $G_k$ and $G_k^-$.
\end{enumerate}
\end{remark}

\section{The Distributed LST-based Protocol}
The above centralized protocol is appropriate for unbalanced MANETs which owns many strict assumptions. In this section, we relax these strict assumptions and let all of the nodes in MANETs be homogeneous with limited abilities and resources. Our goal is to design a group key agreement protocol in a general MANETs environment without making additional assumptions on the availability of any supporting infrastructure. We assume that all nodes in MANETs are identical, the transmission range of all nodes is denoted as $d_\mathrm{max}$.
\begin{definition}
(Topology $\overline{G}_k$): The topology,  $\overline{G}_k (k=0,1,\ldots)$, is an undirected graph $\overline{G}_k=(\overline{V}_k,\overline{E}_k)$, where $\overline{V}_k$ is all the nodes in the network after the occurrence of $k$th node event, and $\overline{E}_k=\{(v_i,v_j): d_{ij}\leq d_{\mathrm{max}},v_i,v_j\in \overline{V}_k\}$  is the edge set of $\overline{G}_k$. In particular, $\overline{G}_0$ is the initial network topology with no node event occurs.
\end{definition}

\begin{definition}
(Neighborhood Subgraph): $N^i(\overline{G}_k)$ is the set of nodes that node $v_i$ has directly connected relationship in $\overline{G}_k$, i.e., $N^i(\overline{G}_k)=\{v_j\in \overline{V}_k:d_{ij}\leq d_{\mathrm{max}}\}$. For an arbitrary node $v_i\in \overline{V}_k$, let $\overline{G}^i_k=(\overline{V}^i_k,\overline{E}^i_k)$  be the induced subgraph of $\overline{G}_k$ such that $\overline{V}^i_k=N^i(\overline{G}_k)$. $\overline{G}^i_k=(\overline{V}^i_k,\overline{E}^i_k)$ is called the neighborhood subgraph of node $v_i$.
\end{definition}

For a node $v_i$, in its neighborhood subgraph $\overline{G}_k^i$, it satisfies Assumption \emph{A}1 that we formalized for the centralized protocol. Because of this, each node can be treated as a \emph{local leader node} in its neighborhood subgraph. Our solution, inspired by \cite{li2003design}, is that each node constructs and maintains its LST in its neighborhood subgraph. Let $LST^i_k$ denote the constructed LST in $\overline{G}_k^i$. Then, the superposition of $LST_k^i (i=1,\ldots,n)$ forms the topology $\overline{G}_k^+$.

\begin{definition}
(Topology $\overline{G}_k^+$): The topology $\overline{G}_k^+(k=0,1,\ldots)$, is a subgraph of $\overline{G}_k$, $\overline{G}_k^+=(\overline{V}^+_k,\overline{E}^+_k)$, where $\overline{V}^+_k=\overline{V}_k$, and $\overline{E}^+_k$ denotes the superposition of edges in $LST_k^i (i=1,\ldots,n)$ for all nodes in $\overline{G}_k$. The word "superposition" means that an edge $e\in \overline{E}_k^+$, if and only if $e\in LST_k^i$,~~$\exists i\in\{1,\ldots,n\}$.
\end{definition}

\noindent\rule{255pt}{1pt}

\textbf{The Distributed Protocol}

\vspace{-2mm}

\noindent\rule{255pt}{1pt}
\vspace{-5mm}
\begin{enumerate}[Step 1:]

\item After the occurrence of $k$th node event $\eta_k (k=0,1,\ldots)$, each local leader node locally broadcasts a \texttt{Hello Message}. The \texttt{Hello Message} should contain its unique ID. In this way, the leader node constructs the topology of its neighborhood subgraph;
\item Each local leader node locally broadcasts a \texttt{Weight Message}, which is subsequently applied to calculate the weight of each edge in its neighborhood subgraph;

\item Each local leader node invokes \textbf{Procedure T-M} to maintain a LST, the LST in $\overline{G}_k^i$  is denoted as  $LST_k^i$. Then, the superposition of  $LST_k^i (i=1,\ldots,n)$ forms the topology $\overline{G}_k^+$;

\item	Construct $\overline{G}_k^+$ in the network. The construction process is similar to \textbf{Procedure S-C};

\item A randomly selected node generates a new session key and distributes it based on $\overline{G}_k^+$. The distribution process is similar to \textbf{Procedure S-K-D}.
\end{enumerate}
\vspace{-3mm}
\noindent\rule{255pt}{1pt}

\begin{remark}
Compared with the centralized protocol, there are several subtle differences in the distributed one.
\begin{itemize}

\item In Step 3, the original one global leader node is replaced by several local leader nodes. The extended Kruskal algorithm is performed by several local leader nodes in a distributed way. Consequently, a node event may cause several local leader nodes to reconstruct their LSTs in the distributed protocol, whereas there is only one reconstruction of the ST in the centralized protocol.
\item In Step 3, the topology $\overline{G}_k^+$ may not be a ST. In general, $\overline{G}_k^+$ is a connected graph with several redundant secure links. See Theorem \ref{theorem10} for the connectivity proof.
\item In Step 5, it is clear that a connected graph $G_k^+$  can also ensure the distribution of session key. If a node receives several encrypted session key messages from different nodes, it only forwards the session key the first time it receives the message.
\end{itemize}
\end{remark}

The verbal description of the above protocol is as follows.
After the occurrence of a node event, each local leader node in the network reconstructs a LST in its neighborhood subgraph independently using the extended Kruskal algorithm which is similar to the centralized protocol. Once $LST_k^i (i=1,2,\ldots)$ are reconstructed, a new connected graph $\overline{G}_k^+$  is derived based on which a new session key is distributed via secure links.

\section{Properties of the Topology $\overline{G}_k^+$}

\begin{figure*}[t]
\centering
\subfigure[The network topology $G_k^+$, $d_{\mathrm{max}}\!=\!4$]{
\label{figa}
\includegraphics[width=7 cm]{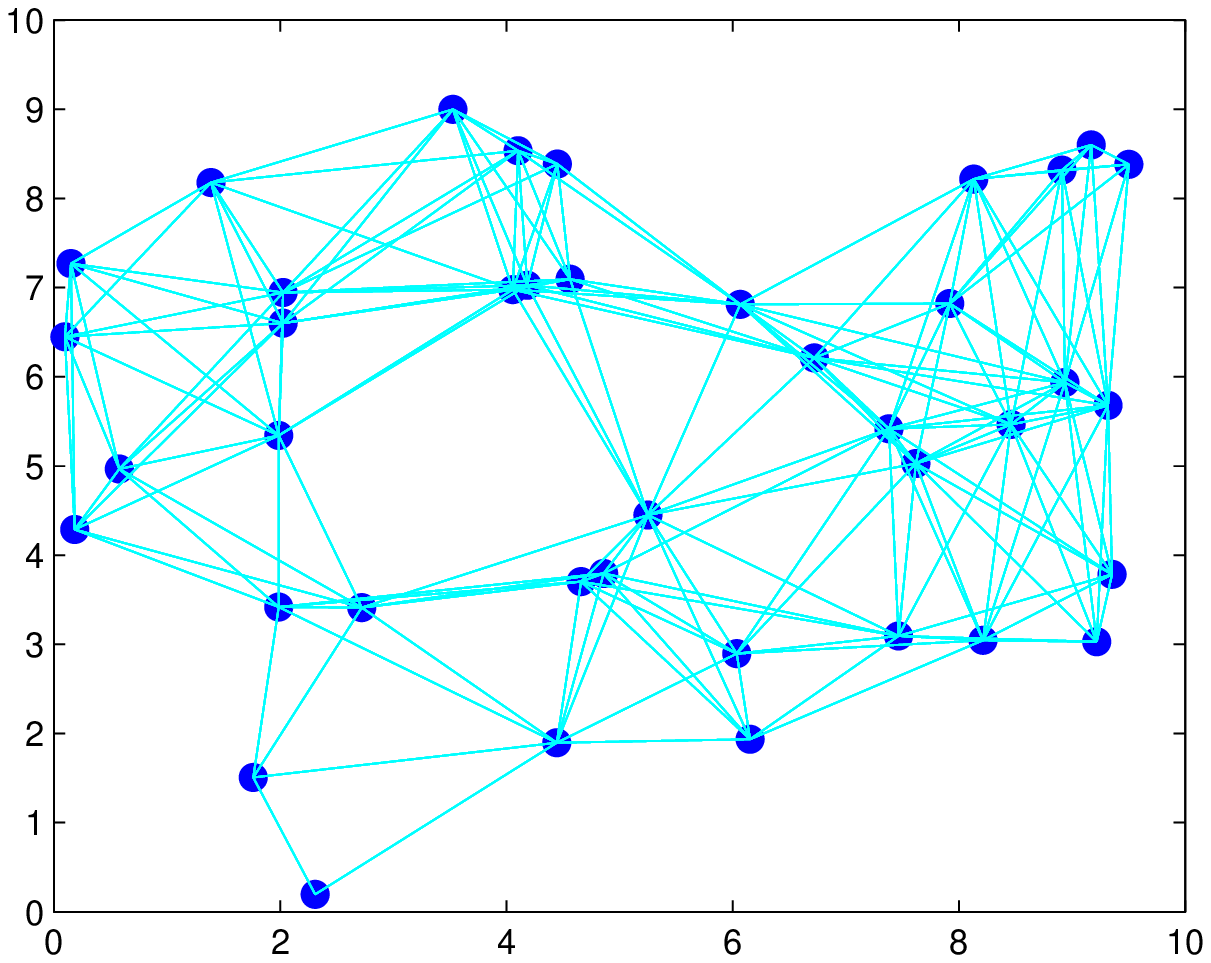}
}~~~~~
\subfigure[The corresponding $\overline{G}_k^+$  derived under LMST, $d_{\mathrm{max}}\!\!=\!\!4$]{
\label{figb}
\includegraphics[width=7 cm]{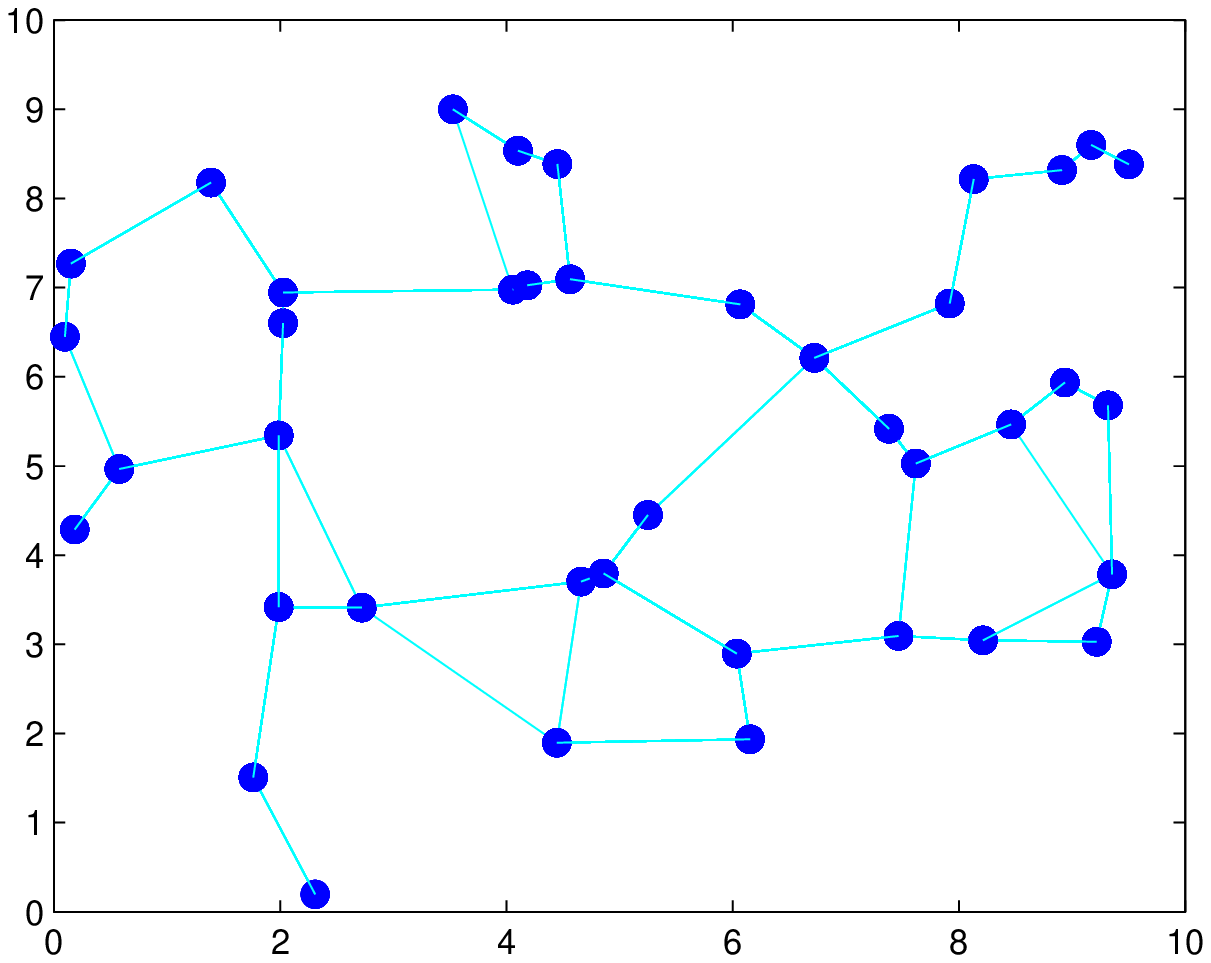}
}
\subfigure[The network topology $G_k^+$, $d_{\mathrm{max}}\!=\!7$]{
\label{figc}
\includegraphics[width=7 cm]{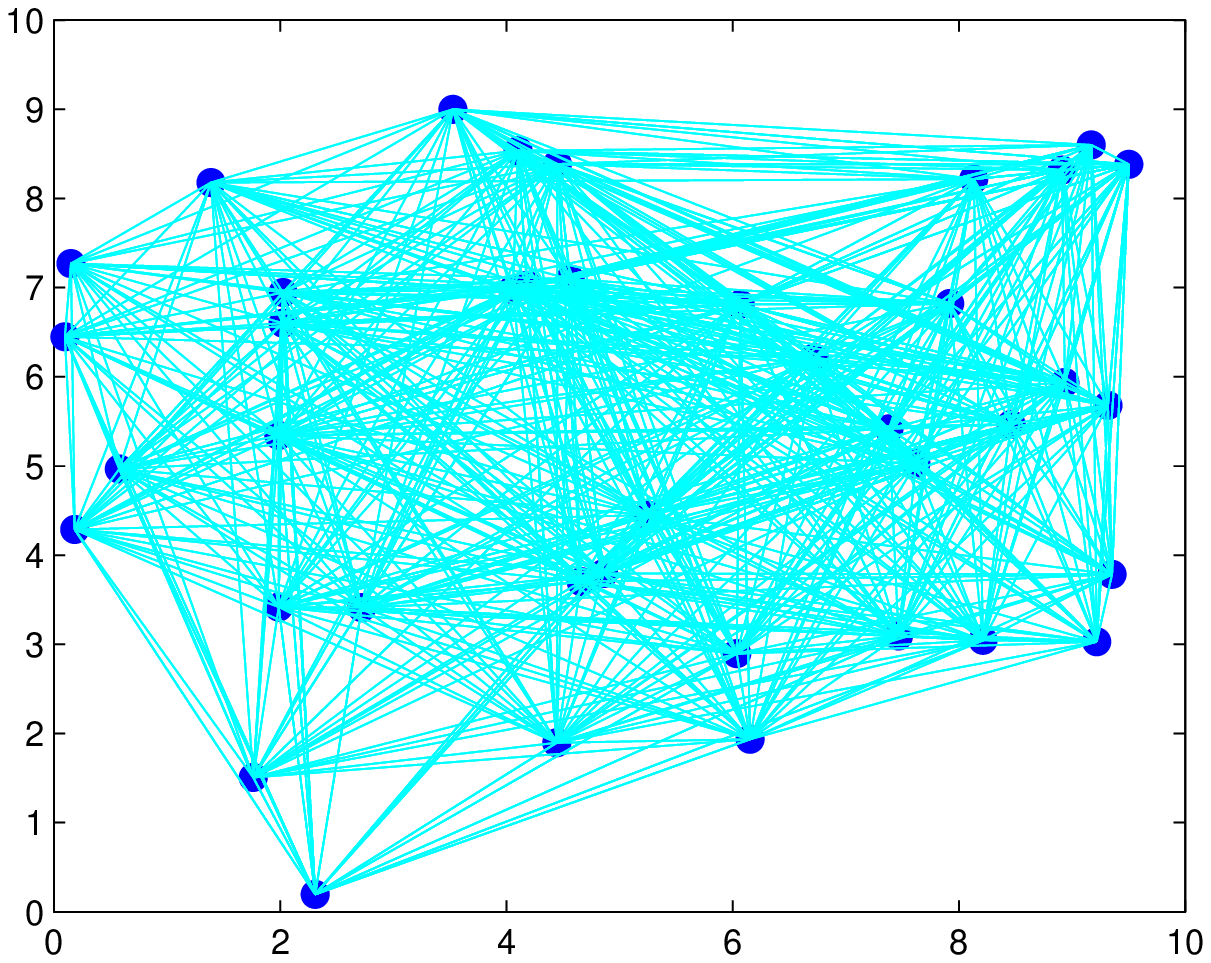}
}~~~~~
\subfigure[The corresponding $\overline{G}_k^+$  derived under LMST, $d_{\mathrm{max}}\!\!=\!\!7$]{
\label{figd}
\includegraphics[width=7 cm]{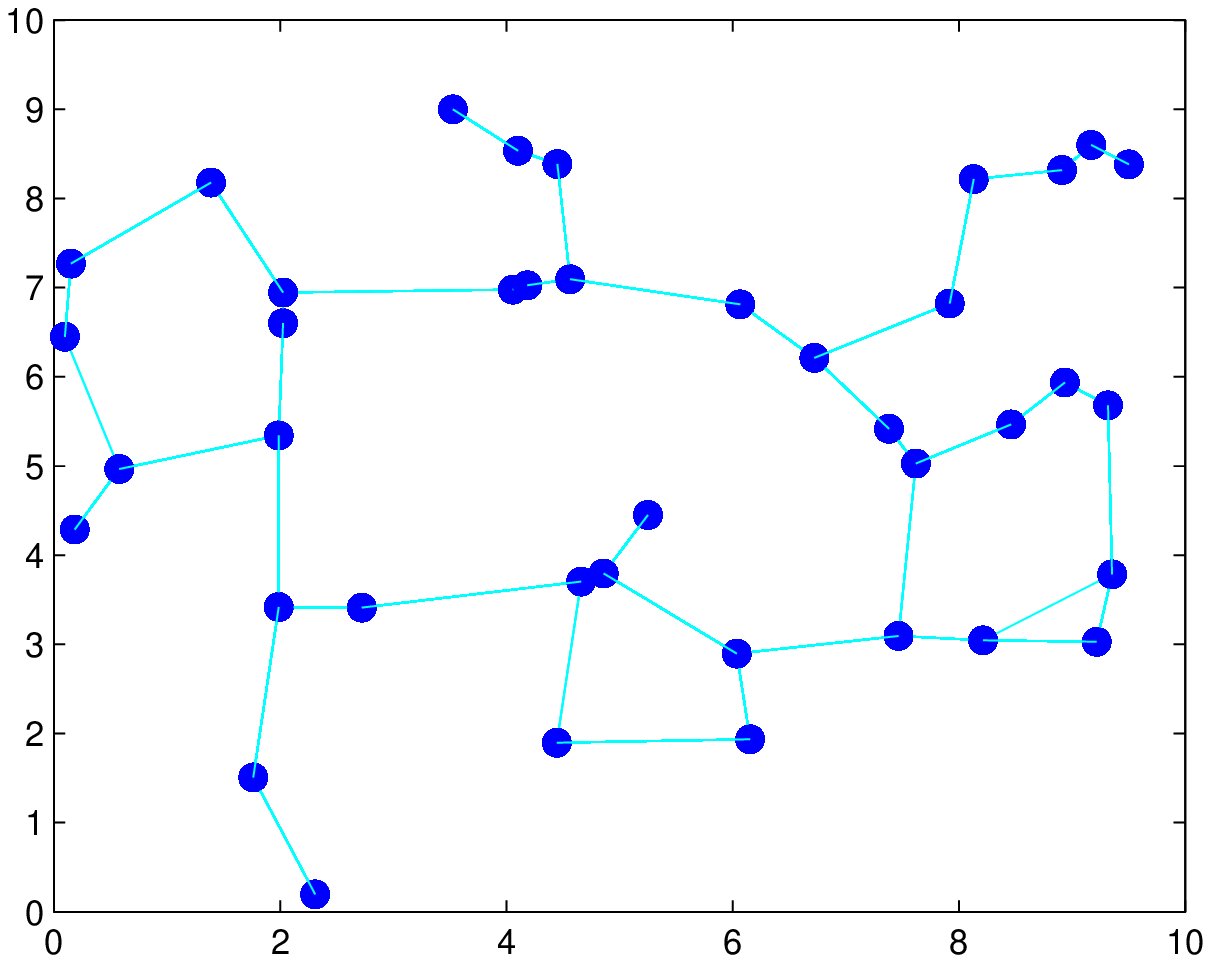}
}
\subfigure[The network topology $G_k^+$, $d_{\mathrm{max}}\!=\!10$]{
\label{fige}
\includegraphics[width=7 cm]{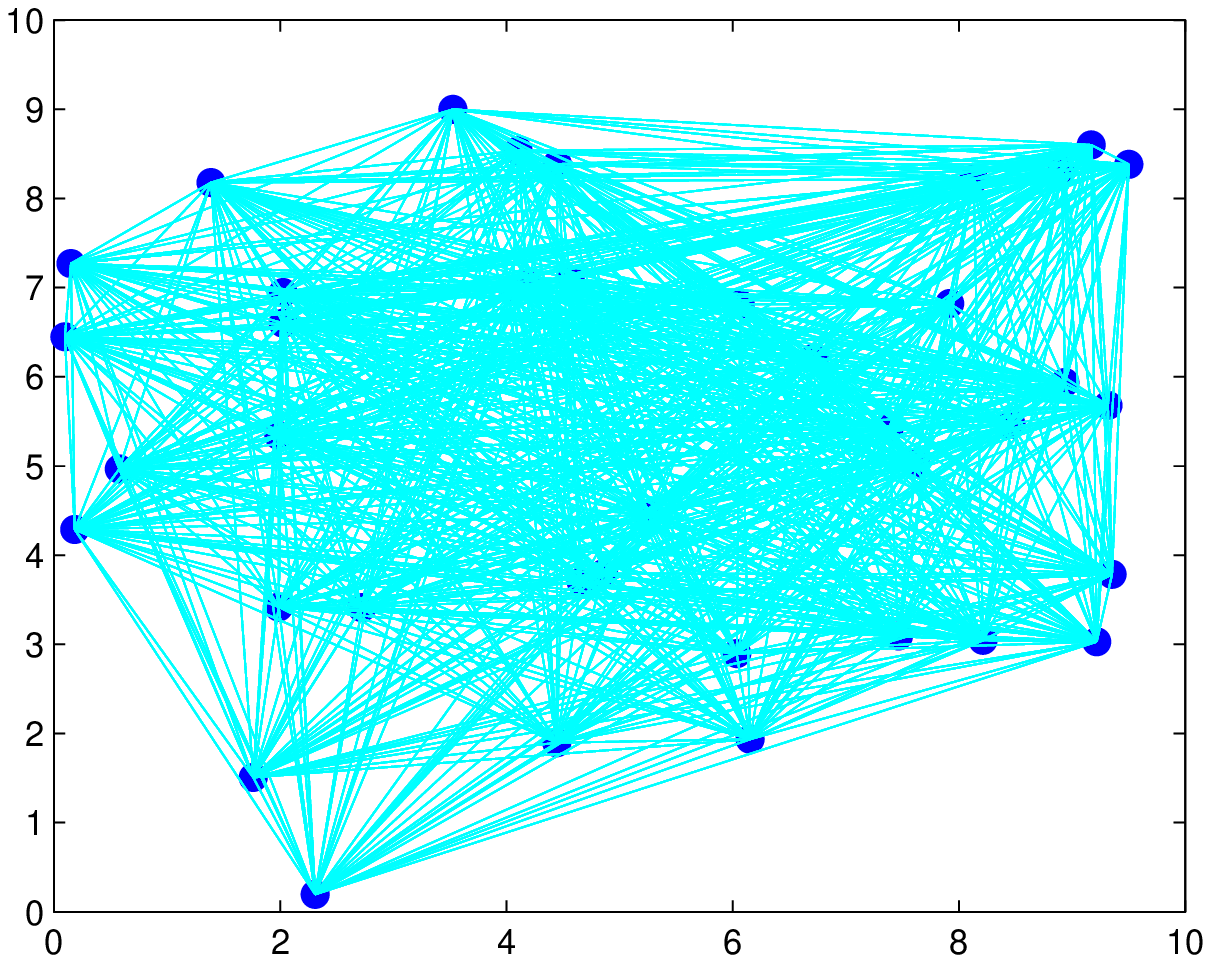}
}~~~~~
\subfigure[The corresponding $\overline{G}_k^+$  derived under LMST, $d_{\mathrm{max}}\!\!=\!\!10$]{
\label{figf}
\includegraphics[width=7 cm]{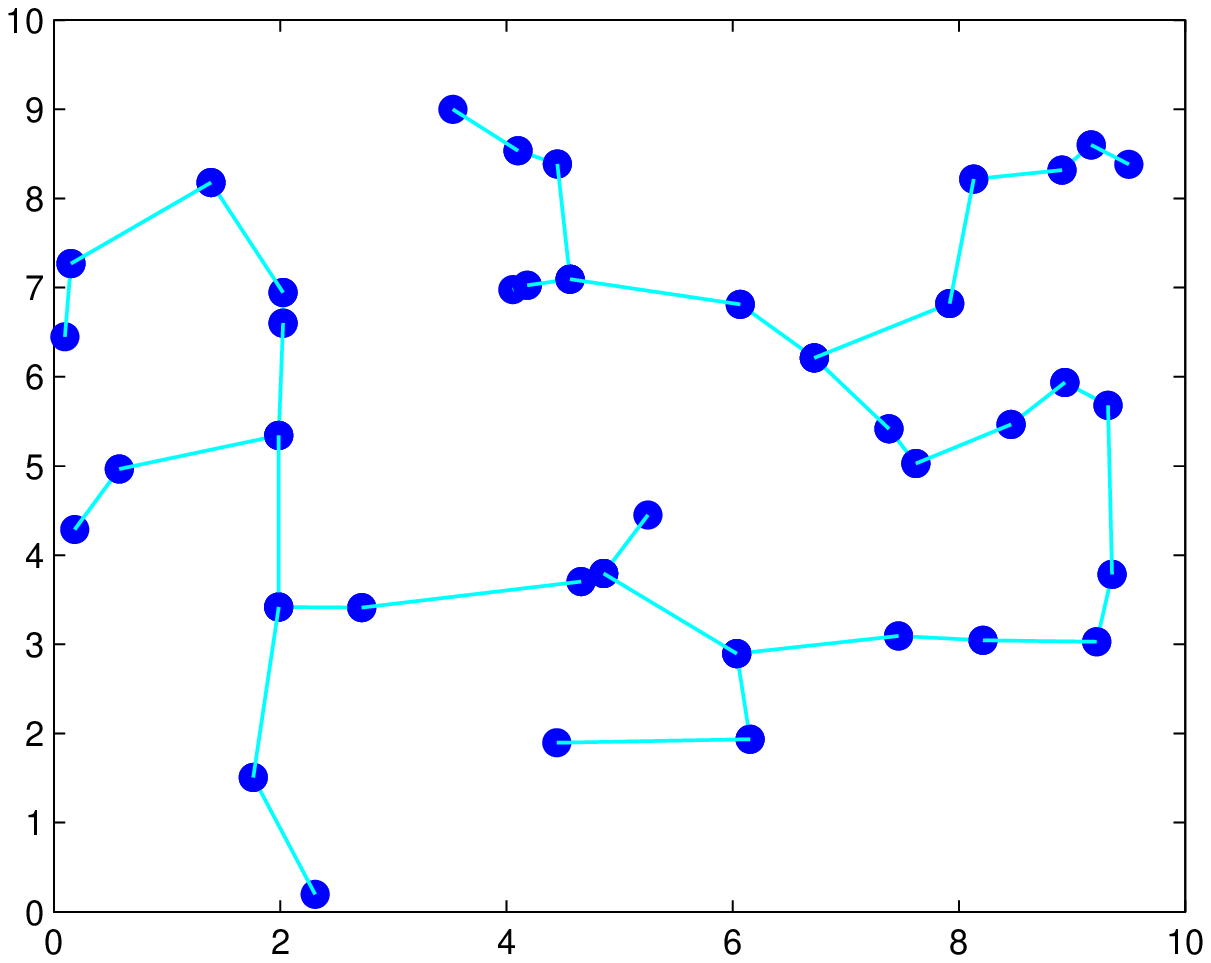}
}

\caption{There are 40 nodes which are generated randomly in a $10\times 10$ square area.}
\label{fig20}
\end{figure*}

In Remark 3, $\overline{G}_k^i$ is claimed to be a connected graph. The proof is given below.

\begin{definition}
(Network Connectivity): For any two nodes $v_i,v_j\in V(G)$, node $v_i$ is said to be connected to node $v_j$  (denoted as $v_i\Leftrightarrow v_j$) if there exists a path ($v_i=w_0,\ldots,w_m=v_j$)  such that $w_j\leftrightarrow w_{j+1}$ , $j=0,\ldots,m-1$, where $w_k\in V(G)$, $k=0,\ldots,m$. It follows that if $u\Leftrightarrow v$ and $u\Leftrightarrow w$, then $v\Leftrightarrow w$.
\end{definition}

\begin{lemma}\label{lemma10}
For any two nodes $v_p,v_q\in \overline{V}_k$, if $v_p,v_q\in LST_k^i$, i.e., $v_p,v_q$ are on an identical ST, then $v_p\Leftrightarrow v_q$.
\end{lemma}

\begin{lemma}\label{lemma20}
For any two nodes $v_p,v_q\in\overline{V}_k$, if $d_{pq}<d_{\mathrm{max}}$, then $v_p\Leftrightarrow v_q$.
\end{lemma}

\begin{proof}
If $d_{pq}<d_{\mathrm{max}}$, it is obvious that $v_p\in N^q(\overline{G}_k)$, then  $v_p,v_q\in LST_k^q$. It holds from Lemma \ref{lemma20} that $v_p\Leftrightarrow v_q$.
\end{proof}

\begin{theorem}\label{theorem10}
$\overline{G}_k^+$ preserves the connectivity of $\overline{G}_k$  , i.e., $\overline{G}_k^+$ is connected if $\overline{G}_k$ is connected.
\end{theorem}

\begin{proof}
Suppose that $\overline{G}_k=(\overline{V}_k,\overline{E}_k)$ is connected. For any two nodes $v_p,v_q\in \overline{V}_k$, there exists at least one path ($v_p=w_0,\ldots,w_m=v_q$) from $v_p$ to $v_q$, where $(w_i,w_{i+1})\in \overline{E}_k$, $i=0,\ldots, m-1$ and $d(w_i,w_{i+1})\leq d_{\mathrm{max}}$. Since $w_i\Leftrightarrow w_{i+1}$, by Lemma \ref{lemma20}, we have $v_p\Leftrightarrow v_q$.
\end{proof}

Theorem \ref{theorem10} ensures that $\overline{G}_k^+$ is a connected graph. However, the topology $\overline{G}_k^+$ may not be a ST, there may be several redundant secure links in $\overline{G}_k^+$.

The superposition of  $LST_k^i (i=1,\ldots,n)$ forms the topology $\overline{G}_k^+$. If all of $LST_k^i(i=1,\ldots,n)$ are LMSTs, we say that the topology $\overline{G}_k^+$ is \emph{derived under LMST}.

\begin{theorem}\label{theorem20}

\rmnum{1})	If network topology $\overline{G}_k$ is a ST, the extended Kruskal algorithm is applied in a distributed way to derive $\overline{G}_k^+$. Then $\overline{G}_k^+$ is exactly a MST in $\overline{G}_k$.

\rmnum{2})	If $d_{\mathrm{max}}\geq \mathrm{max}\{d_{ij}\}$, where $i,j=1,\ldots,n$, or equivalently, the network topology $\overline{G}_k$ is a complete graph, and the corresponding $\overline{G}_k^+$ is derived under LMST, then $\overline{G}_k^+$ is exactly a MST in $\overline{G}_k$.

\end{theorem}
\begin{proof}
For \rmnum{1}), according to Theorem \ref{theorem10}, $\overline{G}_k^+$ preserves the connectivity of $\overline{G}_k$. Since $\overline{G}_k$ is a ST, which yields $\overline{G}_k^+=\overline{G}_k$. It holds that $\overline{G}_k^+$ is exactly the MST in $\overline{G}_k$.

For \rmnum{2}), if the network topology is a complete graph, we have $\overline{G}_k^i=\overline{G}_k$,~$\forall i\in\{1,2,\ldots\}$. Since the corresponding $\overline{G}_k^+$ is derived under LMST, which yields $LST_k^i$ is the MST in $\overline{G}_k$, $\forall i\in \{1,2,\ldots\}$. Since $\overline{G}_k^+$ is the superposition of $LST_k^i (i=1,\ldots,n)$, we have $\overline{G}_k^+$ is exactly a MST in $\overline{G}_k$.
\end{proof}

The second part of Theorem \ref{theorem20} is confirmed by Fig. \ref{figf}. Theorem \ref{theorem20} theoretically answers the question that under what conditions, a MST can be derived in the network topology.
However, in general, the distributed protocol will derive a graph with several redundant edges. Let $d_\mathrm{low}$ be the minimum value of $d_\mathrm{max}$, such that $\overline{G}_k$ is a connected graph. Let $d_\mathrm{upper}=\mathrm{max}\{d_{ij}\}$, where $i,j=1,\ldots,n$. It is shown from simulation that if $\overline{G}_k^+$ is derived under LMST, with the increasing of the transmission range $d_\mathrm{max}\in[d_\mathrm{low}, d_\mathrm{upper}]$, the number of redundant edges in $\overline{G}_k^+$  is first increasing and then decreasing on average, as shown in Fig. \ref{fig20}. A large $d_\mathrm{max}$ generates a dense graph, whereas a small  $d_\mathrm{max}$ generates a sparse graph. More specifically, if $d_\mathrm{max}>d_\mathrm{upper}$, then $\overline{G}_k$ is a complete graph and the corresponding  $\overline{G}_k^+$ is exactly a MST. That is to say, there are no redundant edge in $\overline{G}_k^+$, as shown in Fig. \ref{figf}.

It is evident that the redundant secure links will accelerate the key distribution process, whereas it requires more communication and computation overhead to construct them.
Generally, according to Theorem \ref{theorem20}, $LST_k^i(i=1,\ldots,n)$, where $k=1,2,\ldots$, may not be a LMST. We can expect that the extended Kruskal algorithm, which is a greedy algorithm, can construct \emph{approximate LMST} (ALMST) in each neighborhood subgraph. Consequently, the number of redundant edges in $\overline{G}_k^+$  derived under ALMST is expected to be larger than that derived under LMST on average. It is demonstrated by simulation that if $\overline{G}_k^+$  is derived under LMST, the number of redundant edges in $\overline{G}_k^+$  is relatively small in comparison with the total number of edges in  $\overline{G}_k$. For instance, the total number of edges of  $\overline{G}_k$ in Fig. \ref{figa} is 295, whereas there are only 8 redundant edges in the corresponding $\overline{G}_k^+$, as shown in Fig. \ref{figb}. This fact ensures the high-efficiency of the proposed LST-based distributed protocol.

\section{Protocol Properties Analysis}

The proposed protocols feature the following basic properties.
\begin{itemize}
\item	A hybrid approach, which is the combination of key agreement and key distribution via symmetric encryption, is exploited.
\item	A ST or LSTs are adopted to form a connected network topology.
\item   The extended Kruskal algorithm is employed to handle dynamic events.
\end{itemize}

The group key agreement protocols which own the above three basic properties belong to HSK (H means hybrid, S indicates ST-based, and K implies the usage of the extended Kruskal algorithm) group key agreement scheme. According to the above analysis, our proposed protocols are two typical cases under the HSK scheme.
\section{Informal Security Analysis}

The mobile environment makes the MANETs susceptible to attacks ranging from \emph{passive attacks} to \emph{active attacks}. Passive attacks try to learn information just by listening to the communication of the participants. In active attacks, an attacker attempts to subvert the communication in any way possible: by injecting message, intercepting message, replaying message, altering message, etc..

\subsection{Security against Passive Attacks}

All of the information that is available by passive attackers includes the three cases below.
\begin{enumerate}[\emph{Case} 1:]
\item In \textbf{Procedure S-C}, it is needed to generate secure link key publicly. Thus, the security against passive attacks depends on the selected two-party public key exchange algorithm.
\item In \textbf{Procedure S-K-D}, the transmitted session key is encrypted by a certain symmetric encrypted algorithm using the secure link key. Therefore, the security depends on the selected symmetric algorithm.
\item Other information that transmitted between nodes is useless for passive attacks.
\end{enumerate}

\subsection{Security against Active Attacks}
We have analyzed the security of the proposed protocols under the HSK scheme in an eavesdropper-only model. If the protocols under the HSK scheme are secure in the presence of only passive attackers. Then, we can transform our unauthenticated protocols to secure authenticated protocols by introducing Katz-Yung compiler \cite{katz2003scalable}, see also \cite{dutta2008provably, dutta2011provably} for further details. We will not elaborate it here.

\section{Conclusions}

A centralized protocol and a distributed protocol under the HSK scheme for group key agreement protocol of MANETs have been proposed and analyzed in detail. The extended Kruskal algorithm shows its superiority in achieving the reusability of the precomputed secure links. There are several components in the HSK scheme that should be carefully choose in practical applications, such as the two-party key exchange algorithm and the symmetric encryption algorithm. Remarkably, this paper is not full-fledged. We are going to seek for help to implement the proposed protocols in the real world applications to assess their practical efficiency.

\balance
\bibliographystyle{IEEEtran}
\bibliography{MST}

\ifCLASSOPTIONcaptionsoff
  \newpage
\fi

\end{document}